\documentclass[11pt, letterpaper]{article}
\usepackage[utf8]{inputenc}
\usepackage[english]{babel}
\usepackage{amsfonts}
\usepackage{amssymb}
\usepackage{amsthm}
\usepackage{authblk}
\usepackage{pdfpages}
\usepackage[letterpaper, margin=1in]{geometry}
\usepackage{mathtools}        
\usepackage{wrapfig}
\usepackage{algorithm2e}

\usepackage{tikz}

\usepackage{suffix}           
\usepackage{braket}           

\DeclarePairedDelimiter\ceil{\lceil}{\rceil}
\DeclarePairedDelimiter\floor{\lfloor}{\rfloor}
\DeclarePairedDelimiter\round{[}{]}
\DeclarePairedDelimiter\curly{\{}{\}}
\DeclarePairedDelimiter\card{\vert}{\vert}

\title{Ensemble computation approach to the Hough transform}
\author{Timur M. Khanipov}
\affil{\slshape Institute for Information Transmission Problems\\ \slshape of the Russian Academy of Sciences (IITP RAS)}
\date{February 19, 2018}

\newtheorem{theorem}{Theorem}[section]
\newtheorem{corollary}[theorem]{Corollary}
\newtheorem{proposition}[theorem]{Proposition}
\newtheorem{remark}[theorem]{Remark}
\newtheorem{definition}[theorem]{Definition}
\newtheorem{lemma}[theorem]{Lemma}

\newcommand\defstyle[1]{{\slshape #1}}
\newcommand\supp[1]{\mathrm{supp}\:#1}
\newcommand\orbit[1]{\braket{#1}}
\newcommand\shiftspan[1]{\round{#1}}

\def\assign{\coloneqq}

\newcommand\compchain[1]{\mathfrak{#1}}    

\newcommand\computesby[1]{\stackrel{#1}{\computes}} 
\newcommand\computes{\leadsto}

\newcommand\depth{d}

\newcommand\combws[1]{\card{#1}}   
\newcommand\combw[2]{\combws{{#1}\combines{#2}}}  
\newcommand\ens[1]{\mathcal{#1}}   

\newcommand\Lens{\mathcal L} 

\newcommand\comref{\vee}     
\newcommand\finer{\preccurlyeq}

\newcommand\finerarr{\rightarrow}   

\newcommand\combines{\vartriangleright}   

\newcommand\eqset[3]{X^{{#1}, {#2}}_{#3}}         
\newcommand\eqsetens[2]{\ens X^{{#1}, {#2}}}      
\newcommand\pset[2]{\mathcal{P}^{#1}_{#2}}

\newcommand\lset[2]{\mathcal{L}^{#1}_{#2}}

\newcommand\lsets[1]{\mathcal{L}_{#1}}    
\newcommand\lcount[1]{E_{#1}}   
\newcommand\glcount[1]{E_{#1}^*}   
\newcommand\ecount{E}    

\newcommand\lpat[2]{L^{#1}_{#2}}

\renewcommand\restriction{\big\vert}

\newcommand\dind[2]{\sigma(#1, #2)}  
\newcommand\dinds{\sigma}  

\newcommand\lfuncs[1]{l^{#1}}   

\newcommand\eqdef{\stackrel{\mathrm{def}}{=}}

\newcommand\image{\mathbb I}

\newcommand\w{w}
\newcommand\ww{\w-1}
\newcommand\hh{h-1}
\newcommand\wrange{\curly{0,1,...,\ww}}
\newcommand\X{X}
\newcommand\Y{Y}

\newcommand\modfunc[2][h]{\mathrm{mod}_{#1}(#2)}
\newcommand\modfuncf[2][h]{\mathrm{mod}_{#1}\:#2}   

\newcommand\grf[1]{\hat{#1}}

\begin{document}

\maketitle

\begin{abstract}
It is demonstrated that the classical Hough transform with shift-elevation parametrization of digital straight lines has additive complexity of at most $\mathcal{O}(n^3 / \log n)$ on a $n\times n$ image. The proof is constructive and uses ensemble computation approach to build summation circuits. The proposed method has similarities with the fast Hough transform (FHT) and may be considered a form of the ``divide and conquer'' technique. It is based on the fact that lines with close slopes can be decomposed into common components, allowing generalization for other pattern families. When applied to FHT patterns, the algorithm yields exactly the $\Theta(n^2\log n)$ FHT asymptotics which might suggest that the actual classical Hough transform circuits could smaller size than $\Theta(n^3/ \log n)$.
\end{abstract}


\section{Introduction}

The Hough transform is a well-known procedure in the area of image processing. It is one of discrete analogues of the integral Radon transform and is widely used for solving numerous tasks, the obvious one being line detection. A good (albeit incomplete) review is given in~\cite{priyanka}, it shows that a variety of techniques may be meant under this term. Here we concentrate on its simplest and most straightforward form which we shall call the classical Hough transform (in \cite{priyanka} it would be named the ``standard Hough transform'' (SHT) 
for the case of straight lines). Supposing that an $n\times n$ image is a numerical function $f=f_{x,y}$ on $\mathbb Z^2$ with bounded support $\curly{0, 1, ..., n{-}1}\times\curly{0, 1, ..., n{-}1}$, the classical Hough transform $HT(f)$ maps {\slshape quantizations} $L$ of continuous lines to sums 
\begin{equation}\label{intro_ht_eq}
HT(f)\colon L\mapsto\sum\limits_{(x,y)\in L} f_{x,y}.
\end{equation}
\defstyle{Patterns} $L$, also called \defstyle{digital straight lines} (DSL)~\cite{klette}, are taken from some specific family $\Lens(n)$, the choice of which defines the particular classical Hough transform type. The total number of all possible DSLs is $\Theta(n^4)$~\cite{koplowitz} but in practice $\card{\Lens(n)} = \Theta(n^2)$ is sufficient providing dense enough covering of continuous lines. Since the case of {\slshape mostly vertical} lines is symmetrical to the case of {\slshape mostly horizontal} lines, w.l.o.g. one can assume that line slopes belong to the interval $[0,1]$ (the $[-1,0]$ case is also symmetrical) and consider only the lines which are ``mostly horizontal inclined to the right'' (i.e. all lines are split into four symmetrical ``quadrants''). In this case one of possible $\Lens(n)$ choices would be to take all lines of form
\begin{equation}\label{intro_lines_eq}
\lpat es\colon\;\;\;\;\;\;\;  y = \round*{\frac{e}{n-1}x} + s, \;\;\;\;\; e=0,1,...,n-1, \;\;\; s\in\mathbb Z,
\end{equation}
\begin{wrapfigure}[5]{hR}{0.2\textwidth}
\begin{tikzpicture}
[scale=0.3]
\def\w{11}
\def\e{4}
\pgfmathsetmacro{\q}{\w-1}
\pgfmathsetmacro{\h}{\e+1}

\draw[help lines,step=1] (0,0) grid (\w,\h);

\begin{scope}
[xshift=0.5cm, yshift=0.5cm]

\def\sqpix[#1](#2,#3) {
\draw[color=black,fill=#1] (#2,#3) ++(-0.4,-0.4) rectangle +(0.8,0.8);
}

\def\cipix[#1](#2,#3) {
\fill[color=#1] (#2,#3) circle (0.17);
}

\foreach \x in {0,...,\q}
{
  \pgfmathtruncatemacro{\y}{round(\e * \x / \q) }
  \sqpix[gray!30](\x, \y)
}

\draw (0,0) -- (\q, \e);
\fill (0,0) circle (0.1);
\fill (\q,\e) circle (0.1);
\end{scope}
\end{tikzpicture}
\end{wrapfigure}
i.e. all lines which pass through pairs of integer points $(0,s)$ and $(n-1, s+e)$, lying on continuations of the left and right image borders. For otherwise line~(\ref{intro_lines_eq}) does not intersect with the image, \defstyle{intercept} (or \defstyle{shift}) $s$ can be assumed to be in $(-n,n)\cap\mathbb Z, \;\;\;e$ is called \defstyle{elevation}. This elevation-intercept parametrization produces $\Theta(n^2)$ digital straight lines. In fact such splitting of lines into two families is one way of overcoming the issue of slope unboundedness in~(\ref{intro_lines_eq}) when the $x$ axis intersection angle approaches $\frac{\pi}2$~\cite[section 2.1]{priyanka}. Another possible approach is to use normal line parametrization~\cite{duda}.

Computing the Hough transform~(\ref{intro_ht_eq}) for all patterns~(\ref{intro_lines_eq}) would require $\Theta(n^3)$ additions when a straightforward independent summation along all lines is used. For performance-heavy applications this complexity might be a critical limitation, raising a natural question of reducing the number of binary operations by a careful choice of the summation order. In general form this task is known as the  {\slshape ensemble computation problem}~\cite{gary}, it is more familiar with a different formulation in the boolean circuits theory~\cite{CLB}. The ensemble computation problem is NP-complete \cite{gary} making it extremely hard to devise an optimal algorithm.

One possible workaround is to use a specific approximation to the Hough transform by replacing digital straight lines~(\ref{intro_lines_eq}) with a different ensemble of patterns which would allow recursive computation (see fig.~\ref{fht_fig}). The algorithm constructed in this manner is known (ambiguously) as the fast Hough transform (FHT) and was reinvented several times (\cite{gotz_dis}, \cite{vuillemin}, \cite{gotz}, 
\cite{brady}, \cite{dpn_hough}). Being extremely convenient for computation, it requires only $\Theta(n^2 \log n)$ summations, so certain lack of attention to this method is surprising (for example, published in 2015 survey~\cite{priyanka} does not mention it).

Using a specific result about boolean linear operators complexity~\cite[theorem 3.12]{CLB} it has recently been proved~\cite{rad_lower} that neither the classical~(\ref{intro_lines_eq}) nor the fast Hough transform can be computed in less than $\Theta(n^2\log n)$ additions, but non-trivial ($o(n^3))$ upper bounds were unknown for the classical transform. In this paper we make an improvement in this direction by suggesting a method for building a computational circuit computing the Hough transform in $O(\frac{n^3}{\log n})$ additions. The key idea of the complexity estimation is to combine lines with consequent slopes (elevations) from their common subpatterns, then repeat this for lines with slopes differing by $2$ and so forth finally arriving to the ultimate single elevation. It was inspired by the FHT algorithm but is more sophisticated. The method constructs a tree consisting of image partitions with each parent node being a common refinement of its children.
An interesting fact is that when applied to the FHT patterns, this tree produces exactly the $\Theta(n^2\log n)$ circuit with optimal size, suggesting that the complexity of the classical Hough transform computational circuits produces by the proposed algorithm might be smaller than the proven upper bound.

As in the case of the FHT, we consider {\slshape ``cyclic''} lines from a single quadrant (mostly horizontal inclined to the right). Their ``wrapping over image border'' property is convenient because it provides a fixed line length and guarantees that every pattern is a function graph defined on the whole domain $\curly{0,1,...,n-1}$. The general case is of course reduced to this one, see the discussion section.

The rest of the paper is organized as follows. Section~\ref{facts_section} introduces basic notations and reproduces several useful common facts, section~\ref{big_ensembles_section} establishes the ensembles framework (not connected with images) and provides the method for constructing computational circuits, section~\ref{image_section} introduces an important concept of span partitions and investigates its properties, section~\ref{hough_section} formally defines the Hough ensemble and proves the main complexity estimation theorem, while section~\ref{discussion_section} analyzes the obtained results and suggests a few directions for further research.

\section{Notations and useful common facts}\label{facts_section}

We use the following notations for operating with numbers and sets. For $t\in\mathbb R$ symbols $\floor t$ and $\ceil t$ denote the usual floor and ceiling operations, $\round t\eqdef\floor{t+\frac12}$ is rounding to the nearest integer. For $p\in\mathbb N = \curly{1,2,3,...}$ and $n\in\mathbb Z$ we denote by $\modfunc[p]n$ or $\modfuncf[p]n$ the remainder of dividing $n$ by $p$, satisfying condition $0\leqslant\modfunc[p]n<p$. Symbols $\subset$ and $\supset$ mean {\slshape proper} set inclusion, while $\subseteq$ and $\supseteq$ allow set equality. $A\sqcup B$ means disjoint union, i.e. $C = A\sqcup B$ if $C = A\cup B$ and $A\cap B = \varnothing$. Set cardinality is denoted as $\card A$ and $2^A$ is the set of all $A$ subsets.

Function $f\colon U\to V$ is an \defstyle{injection} if $f(x_1) = f(x_2)$ yields $x_1 = x_2$. Injections have the following important easily verifiable property:
\begin{proposition}\label{injection_prop}
Injective function $f\colon U\to V$ preserves set structure on $U$, i.e. for $R,S\subseteq U$ the following statements are true:
\begin{enumerate}
\item $f(R) \square f(S) = f(R \square S)$, where $\square\in\curly{\cap, \cup, \sqcup, \setminus}$.
\item $f(R) \bigcirc f(S) \Leftrightarrow R \bigcirc S$, where $\bigcirc\in\curly{=, \subset, \supset, \subseteq, \supseteq}$.
\item $\card{f(R)} = \card R$.
\end{enumerate}
\end{proposition}

Any function $f\colon U\to V$ is naturally extended to a function $f_1\colon 2^U\to 2^V$ by rules $f_1(P) \eqdef \curly{f(x)\mid x\in P}\subseteq 2^V$ for $\varnothing\neq P\subseteq U$ and $f_1(\varnothing) \eqdef \varnothing$. It is usually clear which set the argument belongs to, so by standard practice we use the same symbol $f$ for both cases. This extension can further be performed for $2^{2^U}$ and so forth, if $f\colon U\to V$ is injective then all such extensions are also injective.

We actively use the concept of a function graph, so recall the necessary terms. Any function $f\colon U\to V$ induces injective embedding $\grf {f}\colon U\to U\times V$ by the rule $$\grf{f}(x) \eqdef (x, f(x)).$$ A \defstyle{graph} of function $f$ on a subset $A\subseteq U$ is the set $G = \grf{f}(A)\subseteq A\times V$. \defstyle{Projection} $\pi\colon U\times V\to U$ is defined as $\pi((u,v)) \eqdef u$. Restriction $\grf f_0\eqdef \grf f\restriction_A\colon A\to G$ is a bijection with $\grf f_0^{-1} = \pi_0\eqdef\pi\restriction_{G}$. Where it does not lead to confusion, symbols $\grf f$ and $\pi$ are used in place of $\grf f_0$ and $\pi_0$. The following property is obvious yet so useful that we formulate it separately:
\begin{proposition}\label{common_proj_prop}
For two functions $f,g\colon U\to V$ and sets $A, B\subseteq U$ let $G=\grf{f}(A)\cap\grf{g}(B)$. Then the following statements are true:
\begin{enumerate}
\item $f\restriction_{\pi(G)}=g\restriction_{\pi(G)}$.
\item $G = \grf{f}(\pi(G)) = \grf{g}(\pi(G)).$
\item $\pi(G) = A\cap B\cap X^{f,g}$, where $X^{f,g} = \pi(\grf f(U)\cap\grf g(U)) = \curly{x\in U\mid f(x) = g(x)}$.
\end{enumerate}
\end{proposition}

\section{Ensemble computation}\label{big_ensembles_section}
\subsection{Ensembles, partitions and combinations}

Consider some finite set $U$ which we shall call \defstyle{domain}, its subsets will be called \defstyle{patterns}. An \defstyle{ensemble} on $U$ is a {\slshape non-empty} collection $\ens A\subseteq2^U\setminus\curly{\varnothing}$ of {\slshape non-empty} patterns. Pattern $C$ is \defstyle{composed (combined)} of patterns $A$ and $B$ when $C=A\sqcup B$.   \defstyle{Support} $\supp\ens A\eqdef \bigcup\limits_{A\in\ens A}A$. Ensemble $\ens A$ is a $U$-\defstyle{partition}, \defstyle{domain partition} or simply \defstyle{partition}, if $U = \bigsqcup\limits_{A\in\ens A}A$. Partitions have the following obvious property:

\begin{proposition}\label{partition_unique_prop}
If $\ens A$ is a partition and $B = \bigcup\limits_{A\in\sigma} A,\; \sigma\subseteq\ens A$ then $B = \bigsqcup\limits_{A\in\sigma} A$ and this presentation is unique.
\end{proposition}

Partition $\ens A$ is \defstyle{finer} than partition $\ens B$ (and $\ens B$ is \defstyle{coarser} than $\ens A$) if any pattern $A\in\ens A$ is contained in some pattern $B\in\ens B$ (such $B$ is {\slshape unique} because $\ens B$ patterns do not intersect). We also say that $\ens A$ is a \defstyle{refinement} of $\ens B$ and denote this partial order relation between partitions as $\ens A\finer\ens B$ or $\ens A\finerarr\ens B$ (the latter variant for diagrams). This notation tacitly implies that both ensembles $\ens A$ and $\ens B$ are partitions of the appropriate domain (which in this section is $U$). Of course, $\supp\ens A=U$ for any partition $\ens A$ and $\ens A\finer\ens B$ yields $\card A\geqslant\card B$. There is one \defstyle{finest} partition 
\begin{equation}\label{ustar_def}
\ens U = U^*\eqdef \curly{\curly{u}\mid u\in U},
\end{equation}
i.e. $\ens U\finer\ens A$ for any partition $\ens A$, so always 
\begin{equation*}
\card{\ens A}\leqslant\card U=\card{\ens U}.
\end{equation*}
Partition refinement is a \defstyle{fragmentation} and vice versa:

\begin{proposition}\label{finer_partition_decomp_prop}
Let $\ens A$ be a partition. Then ensemble $\ens B$ is a partition and $\ens A\finer\ens B$ iff for any pattern $B\in\ens B$
\begin{equation}\label{finer_partition_decomp_equation}
  B = \bigsqcup\limits_{A\in\alpha(B)} A,
\end{equation}
where $\alpha(B)\eqdef\curly{A\in\ens A\mid A\subseteq B}\subseteq\ens A$ and
\begin{equation}\label{finer_A_decomp_equation}
\ens A = \bigsqcup\limits_{B\in\ens B}\alpha(B).
\end{equation}
\end{proposition}

\begin{proof}
Suppose $\ens B$ is a partition and $\ens A\finer\ens B$. Then for any $A\in\ens A$ there is a unique $B_A\in\ens B$ such that $A\subseteq B_A$ and for all other $B\in\ens B \;\;A\cap B=\varnothing$, so $\alpha(B_1)\cap\alpha(B_2)=\varnothing$ for $B_1\neq B_2$. Any $B\in\ens B$ is decomposed as $B = B\cap U = B\cap\bigsqcup\limits_{A\in\ens A}A = \bigsqcup\limits_{A\in\ens A}(B\cap A) = \bigsqcup\limits_{A\in\alpha(B)} A$
and since any $A\in\ens A$ is contained in some $B$,
$\ens A = \bigcup\limits_{B\in\ens B}\alpha(B) = \bigsqcup\limits_{B\in\ens B}\alpha(B).$
The converse implication proof is even easier.
\end{proof}


Partition $\ens C$ is a \defstyle{common refinement} of partitions $\ens A$ and $\ens B$ if $\ens C\finer\ens A$ and $\ens C\finer\ens B$. Of course, $\ens U$ is always a common refinement for any two partitions but there is always a unique {\slshape coarsest} common refinement 

\begin{equation*}
\ens A\comref\ens B \eqdef \curly{A\cap B\mid A\in\ens A, B\in\ens B}\setminus\curly{\varnothing}\finer\ens A, \ens B.
\end{equation*}

The $\comref$ operation is obviously associative and commutative and $\ens A\finer\ens B$ yields $\ens A\comref\ens B = \ens A$. Directly from the definition, 
\begin{equation}\label{comref_card_ineq}
\card{\ens A\comref\ens B}\leqslant \card{\ens A}\cdot\card{\ens B}.
\end{equation}

We say that ensemble $\ens A$ \defstyle{combines} ensemble $\ens{B}$ ($\ens A, \ens B$ need not necessarily be partitions) if any pattern $B\in\ens B$ can be combined (possibly in multiple ways) using {\slshape some} patterns from $\ens A$: $B = \bigsqcup\limits_{A\in\sigma(B)}A,\; \sigma(B)\subseteq\ens A$. We denote this relation between ensembles as $\ens A\combines\ens B$. Obviously, $\ens U\combines\ens A$ for any ensemble $\ens A$. The combination relation is reflexive and transitive: $\ens A\combines\ens A$ and from $\ens A\combines\ens B,\;\ens B\combines\ens C$ ($\ens A\combines\ens B\combines\ens C$ for short) follows $\ens A\combines\ens C$. However, it is not antisymmetric and thus is not a partial order: $\ens A=\curly{\curly u, \curly v, \curly{u, v}},\;\ens B=\curly{\curly u, \curly v}$ is an example of two ensembles which combine each other yet differ.


If $\ens A\combines\ens B$ then for every pattern $B\in\ens B$ we define its \defstyle{combination weight} $\omega_{\ens A}(B)$ with respect to $\ens A$ as the minimal number of binary $\sqcup$ operations needed to construct $B$ from $\ens A$ patterns:
\begin{equation}\label{set_weight_def}
\omega_{\ens A}(B) \:\eqdef \!\!\!\!\min_{\substack{\sigma\subseteq\ens A:\\
\bigsqcup\limits_{A\in\sigma} A = B}}{\card{\sigma}}-1.
\end{equation}

Consider a \defstyle{binary combination tree} of pattern $B\in\ens B$ with its $n=\omega_{\ens A}(B)+1$ leaves being patterns from one of $\sigma_0\subseteq\ens A$ minimizing the expression above and all levels full except maybe the last one. The depth of such tree is obviously $\ceil{\log_2 n}$. If all its non-leaf nodes are associated with the binary $\sqcup$ operation, such tree represents a possible way of building $B$ using binary disjoint unions and has the minimal depth among all binary trees combining $B$. This justifies defining pattern \defstyle{combination depth} as
\begin{equation}\label{pattern_depth_def_eq}
  \depth_{\ens A}(B) \eqdef \ceil{\log_2 ( \omega_{\ens A}(B) + 1)}.
\end{equation}

\defstyle{Combination weight} $\omega_{\ens A}(\ens B)$ and \defstyle{depth} $\depth_{\ens A}(\ens B)$ of ensemble $\ens B$ with respect to $\ens A$ are
$$\omega_{\ens A}(\ens B) = \combw{\ens A}{\ens B} \eqdef \sum\limits_{B\in\ens B}\omega_{\ens A}(B),$$

$$\depth_{\ens A}(\ens B) \:\eqdef \max_{B\in\ens B}\depth_{\ens A}(B).$$

$\omega_{\ens A}(\ens B)$ is the minimal number of binary $\sqcup$ operations needed to assemble all $\ens B$ patterns from ensemble $\ens A$ patterns {\slshape directly}, i.e. without composing and reusing intermediate patterns. From combination relation and weight and depth definitions directly follow:

\begin{proposition}
Any ensemble $\ens A$ combines itself with $\combw{\ens A}{\ens A} = \depth_{\ens A}(\ens A) = 0.$
\end{proposition}

\begin{proposition}\label{combw_depth_generic_bound_prop}
If $\ens A\combines\ens B$ then
\begin{enumerate}
    \item $\combw{\ens A}{\ens B} \leqslant \card{\ens B}\cdot(\card{\ens A} - 1)$.
    \item For any $B\in\ens B$ its combination weight $\omega_{\ens A}(B)\leqslant \min(\card{\ens A} - 1, \combw{\ens A}{\ens B}).$
    \item $\depth_{\ens A}(\ens B) \leqslant \ceil{\log_2\min(\card{\ens A}, \combw{\ens A}{\ens B}+1)}.$
\end{enumerate}
\end{proposition}

\begin{proposition}\label{two_unions_decomp_prop}
If $\ens A\combines\ens B, \;\ens C\combines\ens D$ and $\ens P = \ens A\cup\ens C, \;\ens Q = \ens B\cup\ens D$ then $\ens P\combines\ens Q$ and

\begin{enumerate}
    \item $\combw{\ens P}{\ens Q}\leqslant \combw{\ens A}{\ens B} + \combw{\ens C}{\ens D}$.
    \item $\depth_{\ens P}(\ens Q)\leqslant \max(\depth_{\ens A}(\ens B), \depth_{\ens C}(\ens D)).$
\end{enumerate}
If (a) $\supp\ens A \cap \supp\ens C = \varnothing$ or (b) $\ens A = \ens C$ and $\ens B\cap\ens D=\varnothing$, then these inequalities transform to equalities.
\end{proposition}

\begin{corollary}\label{ensemble_union_part_corollary}
If $\ens A_i, \ens B_i$ are ensembles on domains $U_i$ and each $\ens A_i\combines\ens B_i$ then for ensembles $\ens A = \bigcup\limits_i\ens A_i, \;\ens B = \bigcup\limits_i\ens B_i$ on domain $U=\bigcup\limits_i U_i$:
\begin{enumerate}
    \item $\ens A\combines\ens B$.
    \item $\combw{\ens A}{\ens B} \leqslant \sum\limits_i \combw{\ens A_i}{\ens B_i}$.
    \item $\depth_{\ens A}(\ens B) \leqslant \max\limits_i\depth_{\ens A_i}(\ens B_i)$.
\end{enumerate}
If (a) domains $U_i$ do not intersect pairwise or (b) all $\ens A_i=\ens A$ and $\ens B = \bigsqcup\limits_i\ens B_i$, then the inequalities transform to equalities.
\end{corollary}

\noindent The refinement partial order is {\slshape stronger} than the combination relation:
\begin{proposition}\label{refinement_comb_prop}
If $\ens A\finer\ens B$ then $\ens A\combines\ens B$ and $\combw{\ens A}{\ens B} = \card{\ens A} - \card{\ens B}$.
\end{proposition}
\begin{proof}
$\ens A\combines\ens B$ follows from~(\ref{finer_partition_decomp_equation}). By proposition~\ref{partition_unique_prop} with $\alpha(B)$ from proposition~\ref{finer_partition_decomp_prop}, $\omega_{\ens A}(B) = \card{\alpha(B)}$ - 1 for 
$B\in\ens B$, so using~(\ref{finer_A_decomp_equation}), $\combw{\ens A}{\ens B} = \sum\limits_{B\in\ens B}(\card{\alpha(B)} - 1) = \card{\ens A} - \card{\ens B}$.
\end{proof}

Combination and refinement relations as well as the corresponding weights and depths are preserved by injections:
\begin{proposition}\label{ensemble_graph_prop}
If $f\colon U\to V$ is an injection then
\begin{enumerate}
    \item $\ens A\combines\ens B$ yields $f(\ens A)\combines f(\ens B)$ with $\combw{f(\ens A)}{f(\ens B)} = \combw{\ens A}{\ens B}$ and $\depth_{f(\ens A)}(f(\ens B)) = \depth_{\ens A}(\ens B)$.
    \item If $\ens A$ is a $U$-partition then $f(\ens A)$ is an $f(U)$-partition and $\ens A\finer\ens B$ yields $f(\ens A)\finer f(\ens B)$.
\end{enumerate}
For both cases decomposition structure is preserved, i.e. for every $B\in\ens B$
$$f(B) = \bigsqcup\limits_{A\in\gamma}f(A)\Leftrightarrow B = \bigsqcup\limits_{A\in\gamma}A.$$
\end{proposition}
\begin{proof}
Follows from proposition~\ref{injection_prop}.
\end{proof}

Obviously, unions of partitions on non-intersecting domains (such unions are of course disjoint) are again partitions and refinement relation is preserved, so from corollary~\ref{ensemble_union_part_corollary} follows
\begin{proposition}\label{partition_union_prop}
If $\ens A_i, \ens B_i$ are $U_i$-partitions each $\ens A_i\finer\ens B_i$ and domains $U_i$ do not intersect pairwise then
\begin{enumerate}
    \item Ensembles $\ens A = \bigsqcup\limits_i \ens A_i$ and $\ens B = \bigsqcup\limits_i \ens B_i$ are $U$-partitions with $U\eqdef\bigsqcup\limits_i U_i$.
    \item $\ens A\finer\ens B$.
    \item $\combw{\ens A}{\ens B} = \sum\limits_i \combw{\ens A_i}{\ens B_i}$.
    \item $\depth_{\ens A}(\ens B) = \max\limits_i\depth_{\ens A_i}(\ens B_i)$.
\end{enumerate}
\end{proposition}

\subsection{Ensemble computation complexity}\label{ensemble_comp_compl}

\defstyle{Computation chain} of length $n$ is a sequence $\compchain C \eqdef \ens A_0\combines\ens A_1\combines ... \combines\ens A_{n-1}\combines\ens A_n$. We say that $\compchain C$ \defstyle{computes} $\ens A_n$ \defstyle{from} $\ens A_0$ and write $\ens A_0\computesby{\compchain C}\ens A_n$. Computation chain \defstyle{weight} and \defstyle{depth} are defined as
\begin{equation}\label{chain_weight_def}
\omega(\compchain C)
= \combws{\ens A_0\computesby{\compchain C}\ens A_n}
\eqdef\combws{\ens A_0\combines ...\combines\ens A_n}
\eqdef \sum\limits_{i=0}^{n-1}\combw{\ens A_i}{\ens A_{i+1}},
\end{equation}
$$
\depth(\compchain C) \eqdef \sum\limits_{i=0}^{n-1}\depth_{\ens A_i}(\ens A_{i+1}).
$$
Let us say that $\ens A$ \defstyle{computes} $\ens B$ and write $\ens A\computes\ens B$ if $\ens A\computesby{\compchain C}\ens B$ with some chain $\compchain C$.

There is a natural interpretation of the computation relation hence its name. Suppose $U=\curly{u_1, u_2, ...}$ and associate each $u_i\in U$ with a variable containing, for example, integer values. Consider the task of computing $n$ sums 
\begin{equation}\label{sums_task_equation}
s_j =\!\!\sum\limits_{u_i\in A_j}\!\!u_i,\;\;\;\;\;\;j=1,...,n,
\end{equation}
where patterns $A_j\in\ens A,\; \card{\ens A} = n$. What is the minimal binary addition operations count needed? We may assume that $\supp\ens A=U$ (it does not change operations count because all variables from $ U\setminus\supp\ens A$ would remain unused). Computing~(\ref{sums_task_equation}) directly takes $\sum\limits_{j=1}^n(\card{A_j}-1) = \omega_{\ens U}(\ens A)$ operations which corresponds to the {\slshape trivial} chain $\ens U\combines\ens A$, here $\ens U = U^*$ again as in~(\ref{ustar_def}). If, however, we first once calculate certain $u_i$ combinations ($\ens B$) and then reuse them, we might get a smaller operations count corresponding to chain $\ens U\combines\ens B\combines\ens A$. If we also compute and reuse combinations of $\ens B$ patterns,  we might reduce additions number even further ($\ens U\combines\ens B\combines\ens C\combines\ens A$). The ultimate question is to find the computation chain $\compchain C = [\ens U = \ens A_1 \combines ... \combines A_m = \ens A]$ with minimal $\omega(\compchain C)$, suggesting the following

\begin{definition}\label{compl_def}
Computation complexity of ensemble $\ens B$ with respect to ensemble $\ens A$ is the number
$$\mu_{\ens A}(\ens B) = \combws{\ens A\computes\ens B} \eqdef \min\limits_{\ens A\computesby{\compchain C}\ens B}\omega(\compchain C).$$
If $\ens A$ is omitted, then we presume that $\ens A=(\supp \ens B)^*$, i.e. the (``internal'') computation complexity of ensemble $\ens B$ is
$$\mu(\ens B) \eqdef \combws{(\supp\ens A)^*\computes\ens A}.$$
\end{definition}

\begin{proposition}\label{subpartition_comp_prop}
Suppose $\ens B$ is a $U$-partition and some ensemble $\ens A\combines \ens B$. Then there is a partition $\ens A_0\subseteq\ens A$ such that $\ens A_0\finer\ens B$ and $\combw{\ens A_0}{\ens B} = \combw{\ens A}{\ens B}, \; \depth_{\ens A_0}(\ens B) = \depth_{\ens A}(\ens B).$
\end{proposition}
\begin{proof}
According to~(\ref{set_weight_def}), for every $B\in\ens B$ we can choose $\alpha(B)\subseteq\ens A$ such that $\omega_{\ens A}(B) = \card{\alpha(B)} - 1$ and $B = \!\!\!\!\bigsqcup\limits_{A\in\alpha(B)} A$. Since $\ens B$ is a partition, $\ens A_0 = \bigcup\limits_{B\in\ens B}\alpha(B)\subseteq \ens A$ is also a partition and satisfies the desired conditions.
\end{proof}

\begin{proposition}\label{partition_comp_compl_prop}
If partitions $\ens A\finer\ens B$, then $\mu_{\ens A}(\ens B) = \card{\ens A} - \card{\ens B}.$
\end{proposition}
\begin{proof}
Suppose $\mu_{\ens A}(\ens B) = \combws{\ens A\computesby{\compchain C}\ens B}$ for computation chain $\compchain C = [\ens A\combines\ens A_1\combines ... \combines \ens A_m = B]$. We will transform $\compchain C$ to $\compchain C'$ without increasing the computation weight so that each $\compchain C'$ element would be a partition refining the next element. We start from the last pair $\ens A_{m-1}\combines \ens A_m$. By proposition~\ref{subpartition_comp_prop} there is a partition $\ens A_{m-1}'\subseteq\ens A_{m-1}$ such that $\ens A_{m-1}'\finer\ens A_m$ and $\combw{\ens A_{m-1}'}{\ens A_m} = \combw{\ens A_{m-1}}{\ens A_m}$.

Obviously, still $\ens A_{m-2}\combines\ens A_{m-1}'$ with $\combw{\ens A_{m-2}}{\ens A_{m-1}'}\leqslant\combw{\ens A_{m-2}}{\ens A_{m-1}}$, so the computation chain $\compchain C_1$ with $A_{m-1}$ changed to $A_{m-1}'$ is still valid and its weight is not greater than that of $\compchain C$. Continuing in this manner by replacing $\ens A_{m-2}$ and etc., we will get computation chain $\compchain C'$ consisting of partitions with $\omega(\compchain C') \leqslant \mu_{\ens A}(\ens B)$. Since $\compchain C$ has {\slshape minimal} possible weight, $\omega(\compchain C') = \mu_{\ens A}(\ens B)$ The statement then follows from~(\ref{chain_weight_def}) and proposition~\ref{refinement_comb_prop}.
\end{proof}

\subsection{The classical ensemble computation problem and circuits complexity}\label{connection_section}

This section briefly establishes a connection between the the concept of computation complexity introduced in definition~\ref{compl_def}, the classical ensemble computation problem and the concept of additive circuit complexity.

The ensemble computation problem is NP-complete and formulated as follows~\cite{gary}. Given a collection $\ens A$ of $U$ subsets and some number $c\in\mathbb N$, is there a sequence $z_1=x_1\sqcup y_1, \;z_2=x_2\sqcup y_2, ...,\; z_n=x_n\sqcup y_c$ of $n\leqslant c$ disjoint unions, where each $x_i, y_i$ is either $\curly u$ for some $u\in U$ or $z_j$ for some $j<i$, and for any $A\in\ens A$ there is some $z_j=A$? In other words, one asks, whether an ensemble of $\ens A$ can be obtained from $\ens U = \curly{\curly u \mid u\in U}$ using not more than $c$ disjoint union operations. 

An equivalent formulation originates from the {\slshape circuits} theory~\cite{wegener}. Suppose $\ens U=\curly{u_i}$ represents a set of number-valued variables, $\ens A=\curly{A_j}\subseteq 2^U$ and we need to compute all sums~(\ref{sums_task_equation}) using a~\defstyle{circuit}, i.e. an acyclic directed graph which has exactly $\card U$ fanin-0 nodes  $u_i$, exactly $\card{\ens A}$ fanout-0 nodes $A_j$ and other nodes representing the $+$ operation (performing summation over its input edges and distributing the result over its output edges). What is the minimal {\slshape size} of such circuit? When restricted to using only fanin-0,1,2 nodes and after defining the \defstyle{circuit size} to be the number of its fanin-2 nodes, we come to the classical ensemble computation problem.

\begin{remark}
This is one of several possible circuit complexity definitions. One may also count the number of edges and consider unlimited fanin nodes and use different binary operations set~\cite{CLB}.
\end{remark}

It is obvious, that if the complexity defined in section~\ref{ensemble_comp_compl} for ensemble $\ens A\;\; \mu(\ens A) = n$, then the classical ensemble computation problem (and, hence, the corresponding minimal circuit complexity) does not exceed $n$. Indeed, take the computation chain $\compchain C=[\ens U=\ens C_0\combines\ens C_1\combines ...\combines C_d=\ens A]$, such that  $\ens U\computesby{\compchain C}\ens A$ and $\omega(\compchain C)=n$. We should simply transform every segment $\ens C_k\combines \ens C_{k+1}$ of the computation chain into a set of $\card{\ens C_{k+1}}$ {\slshape binary combination trees} computing patterns from $\ens C_{k+1}$ (see the paragraph right above~(\ref{pattern_depth_def_eq})) and then join them for $k=1, 2, ..., d$. The resulting circuit will have exactly $n$ disjoint union nodes by $\mu(\ens A)$ definition.

The reverse is also true. Take the minimal circuit computing $\ens A$. We can always transform it so that the distance (number of edges in a connecting path) any output node and all its input (fanin-0) nodes is {\slshape exactly} $d$, where $d$ is the maximal such distance in the original circuit (its {\slshape depth}), and that the number of fanin-2 nodes does not change. We can then define ensembles $\ens C_k$, consisting of nodes which have distance from the input nodes exactly $k$. Obviously, $\ens U=\ens C_0\combines\ens C_1\combines ... \combines \ens C_d=\ens A$ with the weight equal to the number of fanin-2 nodes.

\subsection{Partition trees}

\begin{proposition}\label{two_part_comp_compl_prop}
Suppose $\ens A$ and $\ens B$ are $U$-partitions. Then
\begin{equation}\label{two_part_comp_compl_ineq}
\mu(\ens A\cup\ens B)\leqslant \mu(\ens A) + \mu(\ens B)-\mu(\ens A\comref\ens B)
= \card{\ens U} + \card{\ens A\comref\ens B} - \card{\ens A} - \card{\ens B}.
\end{equation}
\end{proposition}
\begin{wrapfigure}[4]{R}{0.2\textwidth}
\begin{tikzpicture}
[%
grow=right,%
level distance=1.2cm,
level 2/.style={sibling distance=0.7cm},%
edge from parent/.style={-latex, draw}%
]%
\node {$\ens U$}%
    child { node {$\ens A\comref\ens B$}%
      child { node {$\ens B$}}%
      child { node {$\ens A$}}%
    };%
\end{tikzpicture}
\end{wrapfigure}

\begin{proof}
Consider computation chain $\compchain C = \ens U\combines\ens A\comref\ens B\combines\ens A\cup\ens B$.
By proposition~\ref{two_unions_decomp_prop}, $\combw{\ens A\comref\ens B}{\ens A\cup\ens B}\leqslant \combw{\ens A\comref\ens B}{\ens A} + \combw{\ens A\comref\ens B}{\ens B} = 2\cdot\card{\ens A\comref\ens B} - \card{\ens A} - \card{\ens B},$ so
$$
\omega(\compchain C) = \combw{\ens U}{\ens A\comref\ens B} + \combw{\ens A\comref\ens B}{\ens A\cup\ens B}\leqslant \card{\ens U} - \card{\ens A\comref\ens B} + 2\cdot\card{\ens A\comref\ens B} - \card{\ens A} - \card{\ens B} = \mu(\ens A) + \mu(\ens B)-\mu(\ens A\comref\ens B),
$$
the last equality follows from proposition~\ref{partition_comp_compl_prop}.
\end{proof}

Proposition~\ref{two_part_comp_compl_prop} shows that the cost of simultaneously computing two partitions $\ens A$ and $\ens B$ is at least by $\mu(\ens A\comref\ens B)$ smaller than the weight of the {\slshape trivial} chain $\ens U\combines\ens A\cup\ens B$, so the smaller $\card{\ens A\comref\ens B}$, the better is the two-step chain $\ens U\combines\ens A\comref\ens B\combines\ens A\cup\ens B$. Given some thought, it seems obvious. Indeed, we can efficiently compute $\ens A$ and $\ens B$ when they have many patterns with big intersections -- we first combine the intersections and then assemble the rest.

Suppose now that our computation target is a union of four partitions: $\ens T = \ens A\cup\ens B\cup\ens C\cup\ens D$. In a similar manner we could arrange $\ens U\computes\ens T$ computation in a hierarchical order (recall that notation $\ens A\finerarr\ens B$ is equivalent to $\ens A\finer\ens B$) in many ways, for example as in fig.~\ref{abctree}.
\begin{figure}[ht!]
\centering
\begin{tikzpicture}
[
grow=right,
level distance = 2cm,
level 1/.style={level distance=2cm},
level 2/.style={sibling distance=0.7cm, level distance=2.5cm},
level 3/.style={sibling distance=1.5cm, level distance=1.5cm},
edge from parent/.style={-latex, draw}
]
\begin{scope}
\node {$\ens U$}
    child { node {$\ens A\comref\ens B\comref\ens C\comref\ens D$}
      child { node {$\ens C\comref\ens D$}
        child { node {$\ens D$} }
        child { node {$\ens C$} }
      }
      child { node {$\ens A\comref\ens B$}
        child { node {$\ens B$} }
        child { node {$\ens A$} }
      }
    };
\end{scope}
\end{tikzpicture}
\caption{One of possible combination orders for $\ens U\computes\ens A\cup\ens B\cup\ens C\cup\ens D$.} \label{abctree}
\end{figure}
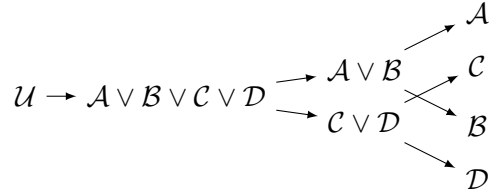
This tree corresponds to chain $\compchain C = \ens U\combines\ens A\comref\ens B\comref\ens C\comref\ens D\combines(\ens A\comref\ens B)\cup(\ens C\comref\ens D)\combines\ens T$. Using the same technique as above, based on propositions~\ref{two_unions_decomp_prop} and~\ref{partition_comp_compl_prop}, we get the following bound:
\begin{equation}\label{four_parts_ineq_1}
\mu(\ens T)\leqslant\mu(\ens A) + \mu(\ens B) + \mu(\ens C) + \mu(\ens D) - \mu(\ens A\comref\ens B\comref\ens C\comref\ens D) - \mu(\ens A\comref\ens B) - \mu(\ens C\comref\ens D),
\end{equation}
or, using cardinalities,
\begin{equation}\label{four_parts_ineq_2}
\mu(\ens T)\leqslant \card{\ens U} + \card{\ens A\comref\ens B\comref\ens C\comref\ens D} + \card{\ens A\comref\ens B} + \card{\ens C\comref\ens D} - \card{\ens A} - \card{\ens B} - \card{\ens C} - \card{\ens D}.
\end{equation}
It means that this particular order is good when each pair $\ens A, \ens B$ and $\ens C, \ens D$ has a lot of well intersecting patterns. In a different case another order might be more efficient. Inserting the $\ens A\comref\ens B\comref\ens C\comref\ens D$ node between $\ens U$ and the rest cannot spoil anything: when $\ens A\comref\ens B\comref\ens C\comref\ens D=\ens U$ (the worst case), the inequality transforms to 
$$
\mu(\ens T)\leqslant 2\card{\ens U} + \card{\ens A\comref\ens B} + \card{\ens C\comref\ens D} - \card{\ens A} - \card{\ens B} - \card{\ens C} - \card{\ens D},
$$
corresponding to a shorter diagram with $\ens A\comref\ens B\comref\ens C\comref\ens D$ replaced by $\ens U$.

\newcommand\thetree{\compchain T}
\newcommand\thechain{c(\thetree)}
These examples can be generalized using the concept of partition trees. Call $\Lens\eqdef\bigcup\limits_i\lsets i$ a \defstyle{partition-union ensemble} if all $\lsets i$ are partitions. A \defstyle{binary partition tree} $\thetree$ is a binary tree rooted at $\ens U$ where each node is a partition which refines its children.  $\thetree$ \defstyle{computes} partition-union ensemble $\Lens$ if its leaves are {\slshape exactly} $\lsets i$, in this case we write $\ens U\computesby{\thetree}\Lens$. No additional restrictions on tree structure are imposed — some nodes might have one child and leaves may reside at different levels. 
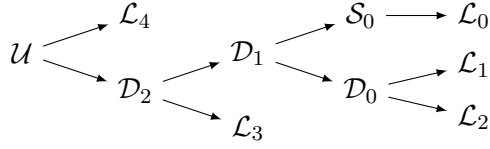
\begin{figure}[ht!]
\centering
\begin{tikzpicture}
[
grow=right,
level 1/.style={level distance=1.5cm, sibling distance=1cm},
level 2/.style={sibling distance=1cm},
level 3/.style={sibling distance=1cm, level distance=1.5cm},
level 4/.style={sibling distance=0.7cm},
edge from parent/.style={-latex, draw}
]
\begin{scope}
\node {$\ens U$}
    child { node {$\ens D_2$}
      child { node {$\lsets3$}
      }
      child { node {$\ens D_1$}
        child { node {$\ens D_0$}
          child { node {$\lsets2$}  }
          child { node {$\lsets1$}  }
        }
        child { node {$\ens S_0$} 
          child { node {$\lsets0$}  }
        }
      }
    }
    child { node {$\lsets4$}
    }
    ;
\end{scope}
\begin{scope}
[xshift=4cm]
\end{scope}
\end{tikzpicture}
\caption{Partition tree of depth $4$ computing five partitions $\lsets i$. Nodes of one type are numbered right to left because we build such trees starting from leaves.} \label{partree}
\end{figure}

\begin{remark}
Any ensemble $\ens A$ can be extended to a partition-union ensemble by adding $\card{\ens A}$ complement patterns $U\setminus A, \; A\in\ens A$. This number can be reduced if we group non-intersecting patterns of $\ens A$.
\end{remark}

Let us define \defstyle{weight of an edge} $\ens A\finerarr\ens B$ as $\combw{\ens A}{\ens B}= \card{\ens A} - \card{\ens B}$ (proposition~\ref{refinement_comb_prop}) and the \defstyle{tree weight} $\omega(\thetree)$ as the sum of all edges weights:
$$
\omega(\thetree) \eqdef \sum\limits_{\ens A\finerarr\ens B}(\card{\ens A} - \card{\ens B}).
$$

Suppose cumulative level ensemble $\ens R_k$ is a union of all partition nodes up to level $k$, so for the example in fig.~\ref{partree}, $\ens R_0 = \ens U, \; \ens R_1 = \ens U\cup\lsets4\cup\ens D_2, \; \ens R_2 = \ens R_1\cup\ens D_1\cup\lsets3, \; \ens R_3 = \ens R_2\cup\ens S_0\cup\ens D_0, \; \ens R_4 = \ens R_3\cup\lsets0\cup\lsets1\cup\lsets2$. The computation chain $\thechain$ \defstyle{associated} with tree $\thetree$ is
\begin{equation*}
  \thechain\eqdef[\ens U=\ens R_0\combines\ens R_1\combines ... \combines \ens R_{d-1}\combines \Lens],
\end{equation*}
where $d$ is tree \defstyle{depth} (number of edges in a longest path). The chain is valid (indeed a computation chain) by corollary~\ref{ensemble_union_part_corollary} applied to every segment $\ens R_{k-1}\combines\ens R_k$. The last segment $\ens R_{d-1}\combines\Lens$ is valid because $\ens R_{d-1}\combines\ens R_d$ and $\Lens\subseteq\ens R_d$, so $\ens R_{d-1}\combines\Lens$ and $\combw{\ens R_{d-1}}{\Lens}=\combw{\ens R_{d-1}}{\ens R_d}$. By this same corollary,
$$
\omega(\thechain)\leqslant\omega(\thetree),
$$
implying
\begin{equation}\label{mu_tree_weight_ineq}
\mu(\Lens)\leqslant\omega(\thetree).
\end{equation}
These facts make the term ``$\thetree$ computes $\Lens$'' more clear.

We define \defstyle{tree computation depth} (not to be confused with {\slshape ordinary} depth $d$) as the depth of its associated chain:
$$
  \depth(\thetree)\eqdef\depth(c(\thetree)).
$$
The tree computation depth may not necessarily be equal to the depth of the {\slshape longest} chain $\ens U\combines\ens P_1\combines ... \combines \ens P_d$. There may be a depth-1 branch $\ens U\combines\ens P$ which has greater computational depth than any other branch because one of the patterns in $\ens P$ needs very many $\ens U$ components to be constructed.

\begin{proposition} If a binary tree $T$ nodes are labeled with numbers and every edge $\alpha\to\beta$ is assigned weight $\beta-\alpha$, then the sum of all edge weights
\begin{equation}\label{tree_weight_eq}
  \sum\limits_{\alpha\to\beta}(\alpha-\beta) = \rho + \sum\limits_j \delta_j - \sum\limits_i\lambda_i,
\end{equation}
where $\rho$ is the root label, $\lambda_i$ are leaves labels and $\delta_j$ are labels of nodes with exactly two children.
\end{proposition}
\begin{proof}
For convenience we denote nodes and the corresponding labels by the same symbol. 
Let us use induction over tree depth $d$. Induction base $d=0$ is obvious, both sides of the equality are zero. Suppose now it is valid for all trees with depth $d$ and consider a tree with depth $d+1$. If its root $\rho$ has only one child $\rho'$, then 
we apply the induction hypothesis to subtree rooted at $\rho'$. It has the same set of leaves and two-children nodes, so the new weight is expressed as $(\rho-\rho') + (\rho'+\sum_j\delta_j - \sum_i\lambda_i)$ which transforms to~(\ref{tree_weight_eq}) for the new root.

Suppose now that root $\rho$ has two children, $\rho'$ and $\rho''$ and the corresponding subtrees have leaves and two-children nodes $\lambda_k', \delta_l'$ and $\lambda_m'', \delta_n''$ resp. By induction hypothesis the new weight is 
$$
(\rho-\rho') + (\rho' + \sum\limits_l\delta_l' - \sum\limits_k\lambda_k') + (\rho-\rho'') + (\rho'' + \sum\limits_n\delta_n'' - \sum\limits_m\lambda_m'') =
$$
$$
\rho + (\rho+\sum\limits_l\delta_l' + \sum\limits_n\delta_n'') - (\sum\limits_k\lambda_k' + \sum\limits_m\lambda_m'').
$$
After noting that $\curly{\lambda_i} = \curly{\lambda_k'}\cup\curly{\lambda_m''}$ and $\curly{\delta_j} = \curly{\rho}\cup\curly{\delta_l'}\cup\curly{\delta_n''}$ we get~(\ref{tree_weight_eq}).
\end{proof}
\begin{proposition}\label{tree_leaves_prop}
For any binary tree the number of leaves is the number of two-children nodes plus one.
\end{proposition}
\begin{proof}Use the same tree depth induction technique as for for~(\ref{tree_weight_eq}).
\end{proof}
\begin{corollary}\label{tree_weight_corollary}
If a binary partition tree $\thetree$ computes partition-union ensemble $\Lens=\bigcup\limits_i \lsets i$ and its nodes (partitions) with exactly two children are $\ens D_j$, then
$$
  \omega(\thetree) = \card{\ens U} + \sum\limits_j \card{\ens D_j} - \sum\limits_i \card{\lsets i} = \sum\limits_i\mu(\lsets i) - \sum\limits_j\mu(\ens D_j).
$$
\end{corollary}\label{parttree_weight_corollary}
\noindent Due to~(\ref{mu_tree_weight_ineq}), this equality generalizes~(\ref{two_part_comp_compl_ineq}), (\ref{four_parts_ineq_1}) and~(\ref{four_parts_ineq_2}).

\begin{proof}
The first equality directly follows from~(\ref{tree_weight_eq}). The second one follows from propositions~\ref{tree_leaves_prop} and~\ref{partition_comp_compl_prop}.
\end{proof}

\subsection{Algorithms for building trees}

Corollary~\ref{parttree_weight_corollary} naturally suggests a greedy approach for building a binary partition tree (and hence the associated computation chain), see algorithm~\ref{greedy_alg} which computes partition-union ensemble $\Lens = \bigcup\limits_{i=0}^{\ecount-1}\lsets i$ consisting of $\ecount$ partitions and hopefully has weight smaller than the trivial chain $\ens U\combines\Lens$ has.
\begin{center}
\begin{algorithm}[H]
\SetAlgoLined
\DontPrintSemicolon
\KwData{partition-union ensemble $\Lens=\cup\lsets i$}
\KwResult{partition tree $\thetree$ such that $\ens U\computesby{\thetree}\Lens$}
 $Q\assign \curly{\lsets0, ..., \lsets{\ecount-1}}$\;
 \While{$\card{Q}>1$}{
  Take {\slshape distinct} $\ens A, \ens B\in M$ with minimal $\card{\ens A\comref\ens B}$\;
  $\ens C\assign\ens A\comref\ens B$\;
  $Q\assign Q\cup\curly{\ens C}\setminus\curly{\ens A, \ens B}$\;
  Create node $\ens C$ with edges $\ens C\finerarr \ens A$ and $\ens C\finerarr\ens B$\;
 }
 Create root $\ens U$ and edge $\ens U\finerarr\ens A_0$ for the single remaining $\ens A_0\in Q$\;
 \BlankLine
 \BlankLine
 \caption{Greedy construction of a partition tree computing $\Lens$.}
 \label{greedy_alg}
\end{algorithm}
\end{center}

One flaw of this algorithm is a lack of tree depth control. Another flaw is its complexity. The most costly procedure here is the computation of $\card Q-2$ common refinement cardinalities $\card{\ens C\comref\ens X}, \; \ens X\in Q$ after removing $\ens A, \ens B$ but before adding $\ens C$. The loop is executed $\ecount-1$ times, so with a straightforward implementation $\Theta(\ecount^2)$ refinements should be computed in total, which might be too much.

Should this be the case or if some {\slshape a priori} information is known about the initial $\lsets i$ partitions, a fixed order based on this information might be used. A trivial order would be to arrange the $\lsets i$ partitions sequentially, refine the consequent pairs and repeat this procedure until a single partition is left, then connect it with $\ens U$. Precisely this method works with the Hough patterns because with the natural elevation numbering consequent lines have small $\card{\ens A\comref\ens B}$ (see~(\ref{hough_consequent_ineq})). See algorithm~\ref{partree_std_alg} and its sample output in fig.~\ref{ltree}. The total number of partition refinements here is $\ecount-1$ which a lot better than~$\Theta(\ecount^2)$.

\newcommand\levsize{prev\_level\_size}
\begin{center}
\begin{algorithm}[H]
\KwData{partition-union ensemble $\Lens=\cup\lsets i$}
\KwResult{depth $\ceil{\log_2 \ecount}+1$ partition tree $\thetree$ such that $\ens U\computesby{\thetree}\Lens$}
\SetAlgoLined
\DontPrintSemicolon
 \For{$i\assign 0$ \KwTo $\ecount-1$} {
   $\lset 0i\assign\lsets i$\;
 }
 $k\assign 1$ \tcp*[r]{currently constructed level}
 $\levsize\assign \ecount$\;
 \While{$\levsize > 1$} {
   $i\assign 0$\;
   
   \While(\tcp*[f]{new nodes with two children}){$2i+1<\levsize$}{
     $\lset ki\assign \lset{k-1}{2i}\comref\lset{k-1}{2i+1}$\;
     Create node $\lset ki$ with edges $\lset ki\finerarr\lset{k-1}{2i}$ and $\lset ki\finerarr\lset{k-1}{2i+1}$\;
     $i\assign i+1$\;
   }

   \If(\tcp*[f]{the last new node has only one child}){$2i<\levsize$} {
     $\lset ki\assign\lset{k-1}{2i}$\;
     Create node $\lset ki$ with edge $\lset ki\finerarr\lset{k-1}{2i}$\;
   }
   $k\assign k+1$\;
   $\levsize\assign i+1$\;
 }
 Create root $\ens U$ and edge $\ens U\finerarr\lset {k-1}0$\;
 \BlankLine
 \BlankLine
 \caption{Building a depth $\ceil{\log_2 {\ecount+1}}$ partition tree computing $\Lens$ by a predefined order.}
 \label{partree_std_alg}
\end{algorithm}
\end{center}

\begin{figure}[ht!]
\centering
\begin{tikzpicture}
[
grow=right,
level 1/.style={level distance=1.5cm, sibling distance=1cm},
level 2/.style={sibling distance=2.4cm},
level 3/.style={sibling distance=1.7cm, level distance=1.5cm},
level 4/.style={sibling distance=0.8cm, level distance=2cm},
edge from parent/.style={-latex, draw}
]
\begin{scope}
\node {$\ens U$}
  child { node {$\lset 30$}
    child { node {$\lset 21$}
      child { node {$\lset 12$}
        child { node {$\lset 05\!=\!\lsets 5$} }
        child { node {$\lset 04\!=\!\lsets 4$} }
      }
    }
    child { node {$\lset 20$}
      child { node {$\lset 11$}
        child { node {$\lset 03\!=\!\lsets 3$} }
        child { node {$\lset 02\!=\!\lsets 2$} }
      }
      child { node {$\lset 10$}
        child { node {$\lset 01\!=\!\lsets 1$} }
        child { node {$\lset 00\!=\!\lsets 0$} }
      }
    }
  };
\end{scope}
\begin{scope}
[xshift=4cm]
\end{scope}
\end{tikzpicture}
\caption{Partition tree constructed by algorithm~\ref{partree_std_alg} to compute ensemble $\lsets0\cup\lsets1\cup\lsets2\cup\lsets3\cup\lsets4\cup\lsets5$.} \label{ltree}
\end{figure}
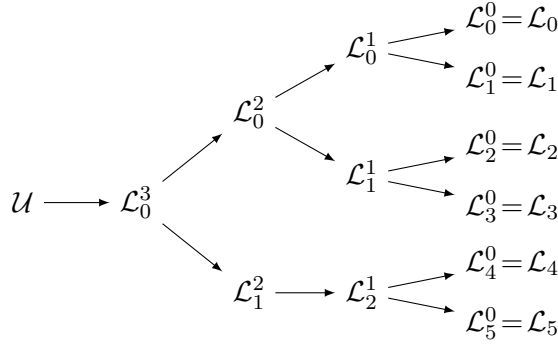

All $\lset ki$ reside on the same level, let us denote their number as $\lcount k$ and the number of $\lset ki$ having two children (i.e. created inside the inner loop) as $\glcount k\leqslant \lcount k$. In the beginning of every iteration of the outer loop $\levsize=\lcount{k-1}$ and the inner loop makes $\glcount k$ iterations. The distribution of single-child nodes depends on the binary representation of $\ecount$ but it is guaranteed that for any $k$ it can only be the last node $\lset ki$ (i.e. with the biggest possible $i$).

\begin{proposition}\label{lglcount_ineq_prop}
$\lcount k = \ceil*{\frac \ecount{2^k}}$ and $\glcount k = \floor*{\frac \ecount{2^k}}$.
\end{proposition}
\begin{proof}
$\lcount k$ and $\glcount k$ obviously satisfy recurrence relations $\lcount {k+1} = \ceil*{\frac{\lcount k}2}$ and $\glcount {k+1} = \floor*{\frac{\glcount k}2}$.

For any $\mathbb R\ni t\geqslant 0$ holds $\ceil*{ \frac {\ceil{t}} 2 }= \ceil*{ \frac t2 }$. Indeed, case $t\in\mathbb Z$ is trivial, for other $t$ use representation $t=\floor t + \lbrace t\rbrace,\;0<\lbrace t\rbrace<1$ and consider cases of $\floor{t}$ being odd or even. The first formula follows from this fact and the recurrence relation by induction. The same proof works for the second one using the equality $\floor*{ \frac {\floor{t}} 2 }= \floor*{ \frac t2 }$.
\end{proof}

\begin{corollary}
The partition tree constructed by algorithm~\ref{partree_std_alg} has depth $\ceil{\log_2\ecount} + 1$.
\end{corollary}

\noindent We can finally formulate the important statement for estimating the Hough patterns complexity.
\begin{lemma}
Suppose the cardinalities of one level nodes $\lset ki$ of the partition tree $\thetree$ constructed by algorithm~\ref{partree_std_alg} for ensemble $\Lens=\bigcup\limits_{i=0}^{\ecount-1}\lsets i$ are bounded by a sequence $a_k$:
$$
  \max\limits_i\card{\lset ki}\leqslant a_k, \;\;\;\;\;\;\;\; k=0,1,..., \ceil{\log_2\ecount},
$$
Then the tree weight
\begin{equation}\label{main_tree_weight_ineq}
\omega(\thetree)\leqslant \card{\ens U} + \ecount \sum\limits_{k=1}^{\ceil{\log_2\ecount}} \frac{a_k}{2^k} -\sum\limits_{i=0}^{\ecount-1}\card{\lsets i},
\end{equation}
and the computational depth
\begin{equation}\label{main_tree_depth_ineq}
\depth(\thetree) \leqslant \log_2\card{U} + \sum\limits_{k=1}^{\ceil{\log_2\ecount}} \log_2 a_k.
\end{equation}
\end{lemma}
\begin{proof}
The first statement follows from corollary~\ref{tree_weight_corollary} and proposition~\ref{lglcount_ineq_prop} after noting that $\glcount k\leqslant \frac{\card{\ecount}}{2^k}$.
The second statement follows from proposition~\ref{combw_depth_generic_bound_prop} applied to all edges $\lset ki\finerarr\lset{k-1}j$ and corollary~\ref{ensemble_union_part_corollary} applied sequentially to consecutive unions of single-level partitions.
\end{proof}

\section{Ensembles and partitions on images}\label{image_section}
\subsection{Image, shifts and spans}

\defstyle{Image} of width $\w\in\mathbb N$ and height $h\in\mathbb N$ is a set $\image\eqdef X\times Y = \curly{ p_{x,y}\eqdef(x,y)\mid x\in\X, y\in\Y}\subset\mathbb R^2$ of $\card \image = \w \cdot h$ elements $p_{x,y}$ called \defstyle{pixels}, where $\X\eqdef\wrange, \;\Y\eqdef\curly{0,1,...,\hh}$. We consider numbers $w$ and $h$ fixed. Image subsets are also called \defstyle{patterns}.
Here and further \defstyle{projection} is a function $\pi\colon\image\to\X, \;\pi(p_{x,y}) \eqdef x$. The finest image partition $\ens I$ of course consists of pixel-singletons: $\ens I\eqdef\curly{\curly{p}\mid p\in\image}$.

Additive commutative group $\mathbb Z$ acts on $\image$ by (vertical) shifts according to the rule $p_{x,y} + s \eqdef p_{x, \modfunc{y+s}}$ for \defstyle{shift} $s\in\mathbb Z$. This action also induces action on patterns: $$P+s \eqdef \curly{p+s\mid p\in P},$$ supposing $\varnothing + s = \varnothing$, and functions $F\colon X\to\image$:
$$
(F+s)(x)\eqdef F(x)+s.
$$
Projection function $\pi$ is obviously shift-invariant: for any $P$ and $s\in\mathbb Z$, $\pi(P+s) = \pi(P).$ For both cases we denote shift \defstyle{orbit} as $$\orbit z \eqdef \curly{z + s\mid s\in\mathbb Z} = \curly{z + s\mid s\in Y}.$$ Shift \defstyle{span} of ensemble $\ens A$ on $\image$ is the set of $A\in\ens A$ patterns shifted to all possible positions (i.e. the union of pattern orbits):
$$
\shiftspan{\ens A}\eqdef\bigcup\limits_{A\in\ens A} \orbit{A} = \curly{A+s\mid A\in\ens A, \;s\in\Y}.
$$
Obviously, $\shiftspan{\curly A} = \orbit A$ for any pattern $A$.

\begin{remark}
As any group action, shifts on $\image$ and $2^{\image}$ are bijections, hence injections and proposition~\ref{injection_prop} statements also hold for shifts (with $f(x) = x + s$).
\end{remark}

\subsection{Function equality sets}

\begin{definition}
For any two functions $f,g\colon X\to Y$ and $n\in\mathbb Z$, $n$-equality set $\eqset fgn$ is defined as
$$
\eqset fgn\eqdef\curly{x\in X\mid \grf g(x) = \grf f(x) + n} = \curly{x\in X\mid\modfunc{g(x)-f(x)-n} = 0}.
$$
$f,g$-equality partition is ensemble 
$$
\eqsetens fg\eqdef \curly{\eqset fgn\mid n\in\dind fg}
$$
with the $f,g$-equality index set
$$
\dind fg\eqdef \curly{n\in Y\mid \eqset fgn \neq \varnothing}.
$$
\end{definition}
This definition can be visualized as cutting graph $\grf g(X)$ by consequent slices $\grf f(X), \; \grf f(X)+1, \;\grf f(X)+2, ...$ and taking projections. By projection shift-invariance property, $\eqset {f+s}{g+s}n=\eqset fgn$. The $\eqsetens fg$ definition is consistent by the following
\begin{proposition}\label{eqsetens_is_partition_prop}
$\eqsetens fg$ is always an $\X$-partition, i.e.
$$
\X = \!\!\!\!\bigsqcup\limits_{n\in\dind fg}\!\!\!\!\eqset fgn.
$$
\end{proposition}
\begin{proof}
For any $x\in X \;\;\grf g(x) = \grf f(x) + n$ with $n=\modfunc{g(x)-f(x)}$, so by definition $x\in\eqset fgn$ implying $\supp \eqsetens fg=X$. Suppose now $x\in\eqset fg{n_1}\cap\eqset fg{n_2}, \; n_1, n_2\in Y$. Then by definition, $\grf f(x)+n_1 = \grf f(x)+n_2$ which is possible only when $n_1=n_2$ if both $n_1, n_2\in Y$.
\end{proof}

\begin{proposition}\label{lfunc_restriction_prop}
Restrictions $g\restriction_B=(f+n)\restriction_B$ for any $B\subseteq\eqset fgn$.
\end{proposition}
\begin{proof}
Immediately follows from $\eqset fgn$ definition.
\end{proof}

\begin{proposition}
For any set $A\subseteq X$
$$g(A) = \!\!\!\!\!\bigsqcup\limits_{n\in\dind fg}\!\!\!f(A\cap\eqset fgn) + n.$$
\end{proposition}
\begin{proof}
Follows from decomposition
$g(A) = g(A\cap\X) = g(A\cap\bigsqcup\limits_{n\in\dinds}\eqset fgn) = \bigsqcup\limits_{n\in\dinds}g(A\cap\eqset fgn)$ by proposition~\ref{lfunc_restriction_prop}, $\dinds=\dind fg$.
\end{proof}
\begin{corollary}
$\grf g(X) = \bigsqcup\limits_{n\in\dind fg} (f(\eqset fgn) + n).$
\end{corollary}

\subsection{Partition spans}

We will now investigate how span operation interacts with domain $X$-partitions via functions $X\to Y$. The first obvious property allows to ``span'' $\image$-partitions from $X$-partitions using a graph function:
\begin{proposition}\label{single_span_func_part_prop}
If $f\colon\X\to\Y$ and $\ens A$ is an $\X$-partition, then 
\begin{enumerate}
\item $\shiftspan{\grf f(\ens A)}$ is an $\image$-partition.
\item $\pi(\shiftspan{\grf f(\ens A)}) = \ens A$.
\item $\card{\shiftspan{\grf f(\ens A)}} = h\cdot\card{\ens A}$.
\end{enumerate}
$\shiftspan{\grf f(\ens A)}$ is a partition span or a span partition and $\grf f$ is a spanning function.
\end{proposition}
\begin{proof} 
By proposition~\ref{ensemble_graph_prop}, \defstyle{slices} $\ens Q_s = \grf f(\ens A)+s, \; s\in Y$ are partitions in $Q_s = \supp \ens Q_s = f(X) + s$. Supports $Q_s$ obviously do not intersect pairwise, so $\bigsqcup\limits_{s\in Y}Q_s = \image$ and statement~1 and 3 follow from proposition~\ref{partition_union_prop}.
Statement~2 follows from $\pi(\ens Q_0) = \pi(\grf f(\ens A)) = \ens A$, shift-invariance of projection and $\ens Q = \bigcup\limits_{s\in Y}\ens Q_s$.
\end{proof}

\begin{proposition}\label{finer_span_part_prop}
If $f\colon\X\to\Y$ and $\X$-partitions $\ens A\finer\ens B$, then $\shiftspan{\grf f(\ens A)}\finer\shiftspan{\grf f(\ens B)}$ and
\begin{enumerate}
\item Decomposition structure on every slice is the same as in $\ens A\finer\ens B$, i.e. $B=\bigsqcup\limits_{A\in\alpha}A$ corresponds to $\grf f(B) + s = \bigsqcup\limits_{A\in\alpha}\grf f(A) + s$ for any shift $s$, $\alpha=\alpha(B)\subseteq\ens A$ is the unique $B$ decomposition from proposition~\ref{finer_partition_decomp_prop}.
\item $\combw{\shiftspan{\grf f(\ens A)}}{\shiftspan{\grf f(\ens B)}} = h\cdot(\card{\ens A} - \card{\ens B}).$
\item $\depth_{\shiftspan{\grf f(\ens A)}}(\shiftspan{\grf f(\ens B)}) = \depth_{\ens A}(\ens B).$
\end{enumerate}
\end{proposition}
\begin{proof}
The fact that $\shiftspan{\grf f(\ens B)}$ is an $X$-partition and the computation weight equality follow from propositions~\ref{refinement_comb_prop} and~\ref{single_span_func_part_prop}. 
Suppose $P\in\shiftspan{\grf f(\ens B)}$, so $P = \grf f(B) + s$ for some $B\in\ens B$ and shift $s$. Since $\ens A\finer\ens B$, we decompose $B = \bigsqcup\limits_{A\in\alpha}A$, so $P = \grf f(\bigsqcup\limits_{A\in\alpha}A) + s = \bigsqcup\limits_{A\in\alpha}(\grf f(A) + s)$, where each component is an element of $\shiftspan{\grf f(\ens A)}$. This proves $\shiftspan{\grf f(\ens A)}\finer\shiftspan{\grf f(\ens B)}$ and provides decomposition structure on slices.
\end{proof}

\begin{proposition}
For any shifts $s_1, s_2$ and sets $A, B\subseteq X$
$$
(\grf f+s_1)(A)\cap(\grf g+s_2)(B) = (\grf f+s_1)(C) = (\grf g+s_2)(C),
$$
where $C = A\cap B\cap\eqset{f}{g}{s_1-s_2}$.
\end{proposition}
\begin{proof}
It immediately follows from proposition~\ref{common_proj_prop} after noticing that $\eqset{f}{g}{s_1-s_2} = \curly{x\in X\mid (\grf f+s_1)(x)=(\grf g+s_2)(x)}$.
\end{proof}

\begin{corollary}\label{part_span_mix_corollary}
For any $X$-partitions $\ens A$ and $\ens B$ 
$$
\shiftspan{\grf f(\ens A)}\comref\shiftspan{\grf g(\ens B)} = \shiftspan{\grf f(\ens C)} = \shiftspan{\grf g(\ens C)},
$$
where $\ens C = \pi(\shiftspan{\grf f(\ens A)}\comref\shiftspan{\grf g(\ens B)}) = \ens A\comref\ens B\comref\eqsetens fg.$
\end{corollary}
This corollary says that the common refinement of two span-partitions is also a span-partition with its spanning function being either of the two spanning functions used. It also shows that to construct this common refinement one may perform a one-dimensional $\comref$ procedure with mixing in an additional component, equality set of the span functions. This is much easier than building a refinement directly in~$\image$. 

The next corollary helps visualize it as well as better understand the nature of $\eqsetens fg$:
\begin{corollary}
$$\orbit{\grf f(X)}\comref\orbit{\grf g(X)} = \shiftspan{\grf f(\eqsetens fg)} = \shiftspan{\grf g(\eqsetens fg)}$$
and
$$
\eqsetens fg = \pi(\orbit{\grf f(X)}\comref\orbit{\grf g(X)}).
$$
\end{corollary}

\section{The Hough ensemble}\label{hough_section}
\subsection{Definition and basic properties}

We now arrived to the primary target of our research -- the Hough patterns. Consider the following \defstyle{base functions} $f_e\colon X\to Y$ with the number $e\in E\eqdef\curly{0,1,2,...,\card E-1}$ called \defstyle{elevation}:
\begin{equation}\label{hough_cyclic_eq}
f_e(x) \eqdef \modfuncf{\round*{\frac e{w-1}x}}.
\end{equation}
The term ``elevation'' originates from the fact that without the modulo-operation the $f_e$ graphs would be {\slshape elevated} by $e$ pixels at the image border comparing to the origin, i.e. they would pass through points $(0,0)$ and $(w-1,e)$.

\begin{definition}
The Hough ensemble or the ensemble of discrete lines is the  partition-union ensemble $\Lens\eqdef \shiftspan{\grf f_e(\curly{X})\mid e\in E} = \bigcup\limits_{e\in E}\orbit{\grf f_e(X)}$.
\end{definition}
By proposition~\ref{single_span_func_part_prop}, every
$$
\lsets e\eqdef \orbit{\grf f_e(X)}
= \shiftspan{\grf f_e(\curly{X})}
$$
is an $\image$-partition of size $h$, consisting of parallel lines with the same elevation. We say that any pattern $\grf f_e(X)+s\in\lsets e$ is a \defstyle{line} with elevation $e$ and \defstyle{shift} $s$. If $\lsets e$ do not intersect pairwise, i.e. all line patterns with various elevations are distinct, then $$\card\Lens = \card E\cdot h.$$ This is always the case when $\card E\leqslant h$ since every value $f_e(w-1)$ would be unique. Anyway, in all cases $\card\Lens \leqslant \card E\cdot h$.

Our point of interest is to find a non-trivial bound for $\mu(\Lens)$. The trivial bound follows from $\ens I\combines\Lens$ by corollary~\ref{ensemble_union_part_corollary} and propositions~\ref{refinement_comb_prop}, \ref{single_span_func_part_prop}:
$$
\mu(\ens L)\leqslant \card E\cdot h\cdot(w-1).
$$
To get a better result we will use one nice property of the Hough patterns -- as elevations difference decreases, lines similarity grows. 

\begin{proposition}
For any two elevations $e_1, e_2$
$$
\card{\eqsetens{e_1}{e_2}}\leqslant |e_1-e_2+1|,
$$
here and further we for short use notations
$$
\eqsetens{e_1}{e_2}\eqdef \eqsetens{f_{e_1}}{f_{e_2}},
\;\;\;\;\;\;\;\eqset{e_1}{e_2}n\eqdef\eqset{f_{e_1}}{f_{e_2}}n.
$$
\end{proposition}
\begin{proof}
Assume that $e_2>e_1$ (always $\card{\eqsetens fg}=\card{\eqsetens gf}$) and denote linear $X\to\mathbb R$ functions $g_e(x) = \frac e \ww x$, so $f_e(x) = \modfuncf{\round{g_e(x)}}$. One can easily see that for any $x\in\X$
$$
  g_{e_1}(x)\leqslant g_{e_2}(x) \leqslant g_{e_1}(x) + e_2 - e_1.
$$

Since $t\leqslant t'$ yields $\round{t} \leqslant \round{t'}$ and $\round{t + s} = \round t + s$ for $t, t'\in\mathbb R, s\in\mathbb Z$ we get
$$
  \round{g_{e_1}(x)} \leqslant \round{g_{e_2}(x)} \leqslant \round{g_{e_1}(x)} + e_2 - e_1.
$$
This implies that for any $x_0\in X, \;\;\round{g_{e_2}(x_0)} = [g_{e_1}(x_0)] + n(x_0)$ and, hence, $f_{e_2}(x_0) = \modfunc{f_{e_1}(x_0)+ n(x_0)}$ where $n(x_0)\in\curly{0, 1, ..., e_2-e_1}$, which by definition means that $x_0\in\eqset{e_1}{e_2}{n(x_0)}$. The proposition statement follows from proposition~\ref{eqsetens_is_partition_prop} and the fact that $n(x_0)$ takes $e_2-e_1+1$ values.
\end{proof}
\begin{corollary}
For any elevation $e>0$ the Hough ensemble has
\begin{equation}\label{hough_consequent_ineq}
\card{\eqsetens{e-1}e\leqslant 2}.
\end{equation}
\end{corollary}
This corollary is illustrated in fig.~\ref{eqsetens_fig}: elevation-5 line (dark circles) is contained in two consecutive elevation-4 lines (variously shaded cells), generating the two-element $X$-partition $\eqsetens45=\curly{\eqset450, \eqset451}$.
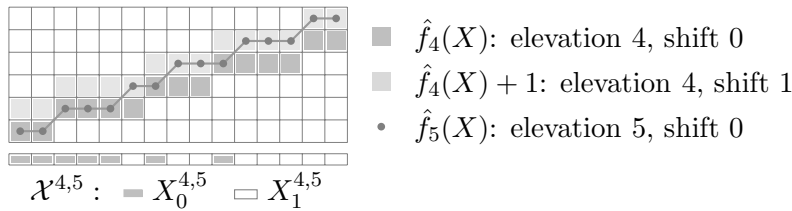
\begin{figure}[ht!]
\centering
\begin{tikzpicture}
[scale=0.3]
\def\w{15}
\def\e{4}

\pgfmathsetmacro{\q}{\w-1}
\pgfmathsetmacro{\h}{\e+3}

\draw[help lines,step=1] (0,0) grid (\w,\h-1);

\draw[help lines,xstep=1, ystep=0.5] (0, -1) grid (\w,-0.5);

\def\sqpix[#1](#2,#3) {
\fill[color=black,fill=#1] (#2,#3) ++(-0.4,-0.4) rectangle +(0.8,0.8);
}

\def\cipix[#1](#2,#3) {
\fill[color=#1] (#2,#3) circle (0.17);
}


\begin{scope}
[xshift=0.5cm, yshift=0.5cm]

\begin{scope}
[yshift=-1.5cm,xshift=1cm]
\draw (1,-1) node {$\eqsetens45:$};
\draw[gray!50,fill=gray!50,xshift=4cm] (-0.4,-1.5+0.1) rectangle (0.4,-1-0.1) node[black,right] {$\eqset450$};

\draw[gray,xshift=9cm] (-0.5,-1.5) rectangle (0.5,-1-0.1) node[black,right] {$\eqset451$};

\end{scope}

\begin{scope}
[yshift=-0.8cm]
\sqpix[gray!50](16,5)
\draw(17,5) node[right] { $\grf f_4(X)$: elevation 4, shift 0};

\sqpix[gray!30](16,3)
\draw(17,3) node[right] { $\grf f_4(X)+1$: elevation 4, shift 1};

\cipix[black!50](16,1)
\draw(17,1) node[right] { $\grf f_5(X)$: elevation 5, shift 0};
\end{scope}

\foreach \x in {0,...,\q}
{
  \pgfmathtruncatemacro{\y}{round(\e * \x / \q) }
  \pgfmathtruncatemacro{\yy}{round((\e+1) * (\x) / \q) }

  \ifnum \y=\yy
    \sqpix[gray!50](\x, \y)
    \sqpix[gray!20](\x, \y+1)

    \fill[gray!50] (\x-0.4,-1.5+0.1) rectangle (\x+0.4,-1-0.1);

  \else
    \sqpix[gray!50](\x, \y)
    \sqpix[gray!20](\x, \y+1)
  \fi
}

\foreach \x in {0,...,\q}
{
  \pgfmathsetmacro{\y}{round((\e+1) * \x / \q) }
  \pgfmathsetmacro{\ynext}{round((\e+1) * (\x+1) / \q) }
  
  \ifnum\x=\q
  \else
    \draw[black!40,thick] (\x, \y) -- (\x+1, \ynext);
  \fi
  \cipix[black!50](\x, \y)
}
\end{scope}

\end{tikzpicture}
\caption{Equality partition for elevations 4 and 5 on an image of width 15.}
\label{eqsetens_fig}
\end{figure}

\subsection{Building the partition tree and $X$ domain reduction}\label{hough_tree_build}
By definition, Hough patterns $\Lens$ is a partition-union ensemble consisting of $\image$-partitions $\lsets e$ and we can apply algorithm~\ref{partree_std_alg} to build the corresponding partition tree. Let's assume that its produced nodes are again denoted as $\lset ki, \;\;k=0, 1,..., \ceil*{\log_2\card E}, \;\; i = 0,1,..., \lcount k = \ceil*{\frac{\card E}{2^k}}$.

We say that $\lset ki$ \defstyle{covers} elevation~$e$ if node $\lsets e=\lset 0e$ resides in the subtree rooted at $\lset ki$. Obviously, for a fixed $k$ nodes $\lset ki$ cover consequent intervals each consisting of $2^k$ consequent numbers, except maybe the last $i=\lcount k-1$ which may cover fewer elevations (only one in the worst case for $\card E=2^n+1,\; n\in\mathbb N$). If we denote the covering set of node $\lset ki$ as $C(k,i)$, then
$$
C(k,i)\eqdef\curly{i2^k,\; i2^k+1,\; i2^k+2, ...,\; i2^k+2^k-1}\cap E.
$$
Of course, for $\lset ki$ having two children, $C(k,i) = C(k-1,2i)\cup C(k-1,2i+1)$ and if $\lset ki$ has one child, then $C(k,i) = C(k-1,2i)$.

By design of the algorithm and the $\comref$ operation associativity and commutativity,
\begin{equation}\label{multiple_vee_presentation}
\lset ki = \!\!\!\!\!\!\bigvee\limits_{e\in C(k,i)}\!\!\!\!\!\lsets e
= \!\!\!\!\!\!\bigvee\limits_{e\in C(k,i)}\!\!\!\shiftspan{\lfuncs e(\curly{X})} = \!\!\!\!\!\!\bigvee\limits_{e\in C(k,i)}\!\!\!\orbit{\lfuncs e(X)}.
\end{equation}
\begin{proposition}\label{pset_lset_connection_prop}
All $\lset ki$ are span-partitions with span functions $\lfuncs e$ for any $e\in C(k,i)$:
$$
\lset ki = \shiftspan{\lfuncs e(\pset ki)},
$$
where
$$
\pset ki\eqdef \pi(\lset ki).
$$
\end{proposition}
\begin{proof}
It follows from~(\ref{multiple_vee_presentation}) by applying corollary~\ref{part_span_mix_corollary} multiple times.
\end{proof}

\begin{proposition}\label{pset_recursive_proposition}
For any $\lset ki$ having to children (i.e. when $i < \glcount k$) and any $e_1\in C(k-1, 2i), \; e_2\in C(k-1, 2i+1)$ holds
\begin{equation}\label{pset_recursive_def_two}
\pset ki = \pi(\lset ki) = \pset{k-1}{2i}\comref\pset{k-1}{2i+1}\comref\eqsetens{e_1}{e_2}.
\end{equation}
If $\lset ki$ had no children then 
\begin{equation}\label{pset_recursive_def_one}
\pset ki=\pset{k-1}{2i}.
\end{equation}
\end{proposition}
\begin{proof}
The second part of the statement is obvious, so let's concentrate on the first one. By two-children assumption, $\lset ki = \lset{k-1}{2i}\comref\lset{k-1}{2i+1} = \shiftspan{\lfuncs {e_1}(\pset{k-1}{2i})}\comref\shiftspan{\lfuncs {e_2}(\pset{k-1}{2i+1})}$ for any $e_1\in C(k-1, 2i),\; e_2\in C(k-1, 2i+1)$ (proposition~\ref{pset_lset_connection_prop}). The first statement then immediately follows from corollary~\ref{part_span_mix_corollary}.
\end{proof}

This proposition means that domain $X$ projection partitions $\pset ki$ can be retrieved independently of $\lset ki$ by assuming $\pset 0e=\curly X$ for all $e\in E$ and using~(\ref{pset_recursive_def_two}) or~(\ref{pset_recursive_def_one}). $\pset ki$ also form a partition tree with the structure identical to the one generated by algorithm~\ref{partree_std_alg}. The $\lset ki$ partitions can in turn be obtained from $\pset ki$ using the spanning procedure. It effectively reduces the complex task of getting a common refinement of two $\lset ki$-sets  to constructing the common refinement of the two corresponding $\pset ki$ sets with one additional component, an equality set, mixed in. The arbitrary choice of $e_1$ and $e_2$ allows us to get a convenient bound on $\card{\pset ki}$.

\begin{proposition}
For all $k$ and $i$
\begin{equation}\label{pset_card_ineq}
\card{\pset ki}\leqslant\min(2^{2^k-1}, w).
\end{equation}
\end{proposition}
\begin{proof}
By proposition~\ref{pset_recursive_proposition}, either
$
\pset ki = \pset{k-1}{2i}\comref\pset{k-1}{2i+1}\comref\eqsetens{e_1}{e_2},
$
with $e_1\in C(k-1, 2i), \; e_2\in C(k-1, 2i+1)$, or
$\pset ki = \pset{k-1}{2i}.$ For the first case we may always choose $e_1$ and $e_2$ so that $e_2=e_1+1$, explicitely $e_1=2i\cdot2^{k-1} + 2^{k-1}-1$ is the greatest element of $C(k-1, 2i)$ and $e_2 = e_1+1$ is the smallest element of $C(k-1, 2i+1)$. With such choice, using~(\ref{comref_card_ineq}), we have $\card{\pset ki}\leqslant2(\max\limits_j\card{\pset{k-1}j})^2$, which is of course also true for the second case $\pset ki = \pset{k-1}{2i}$. Since $\card{\pset0i}=\card{\curly{X}}=1$ for all $i$ and $\eqsetens{e_1}{e_2}\leqslant 2$ by ~(\ref{hough_consequent_ineq}), $\;\card{\pset ki}$ is bounded by sequence $c_k$ defined as $c_0=1, \; c_k=2c_{k-1}^2$. One can verify that $c_k=2^{2^k-1}$. Inequality $\card{\pset ki}\leqslant w$ is trivial, as $\pset ki$ is an $X$-partition.
\end{proof}

\noindent Using proposition~\ref{single_span_func_part_prop} we deduce
\begin{corollary}\label{lset_card_corollary}
For all $k$ and $i$
$$
  \card{\lset ki}\leqslant h\cdot\min(2^{2^k-1}, w).
$$
\end{corollary}

\subsection{Complexity bound}

\begin{lemma}
For any $n\in\mathbb N$
\begin{equation}\label{powerpowersum_ineq}
\sum\limits_{k=1}^n 2^{2^k-k-1}<2^{2^n-n-1}+2^{2^{n-1}-n+1}.
\end{equation}
\end{lemma}
\begin{proof}
We prove this by induction. For $n=1,2,3$ the inequality holds. Suppose that $m>3$ and it is valid for all $n<m$. After denoting the sum on the left side of~(\ref{powerpowersum_ineq}) as $s_n$, the induction hypothesis implies
$$
s_m= s_{m-1} + 2^{2^m-m-1} < 2^{2^{m-1}-m}+2^{2^{m-2}-m+2} + 2^{2^m-m-1}.
$$
Since the second term $2^{2^{m-2}-m+2}<2^{2^{m-1}-m}$ for $m>3$ we get
$$
s_m<2^{2^{m-1}-m}+2^{2^{m-1}-m} + 2^{2^m-m-1} = 2^{2^m-m-1} + 2^{2^{m-1}-m+1}.
$$
\end{proof}

At last we can state the main theorem estimating the Hough ensemble complexity.
\begin{theorem}\label{hough_compl_theorem}
The complexity of the Hough ensemble $\Lens$ produced by $E$ base lines on an image with width $w$ and height $h$ is bounded by the following inequality:
\begin{equation}\label{hough_complexity_ineq}
\mu(\Lens)< \frac{4whE}{\log_2w+1}
\left(1+\sqrt{\frac2w}\right) + h(w-E).
\end{equation}
The corresponding computation chain has depth at most $\log_2 w\cdot(\ceil{\log_2 E} + 1).$
\end{theorem}
\begin{proof}
Suppose we have built the partition tree computing $\Lens$ as described in section~\ref{hough_tree_build}. Because of~(\ref{mu_tree_weight_ineq}) it is enough to estimate the weight of the constructed tree $\thetree$.
\newcommand\tv{\log_2(\log_2 w + 1)}
Denote $t=\tv$ and let $k_0=\ceil t-1$, in this case
$t-1\leqslant k_0 < t$. Taking into account that for the Hough ensemble $\card{\lsets i} = h$ and $\card{\ens I} = w\cdot h$, using corollary~\ref{lset_card_corollary} we can rewrite~(\ref{main_tree_weight_ineq}) as
\begin{equation}\label{tree_assessment_proof_ineq}
\omega(\thetree)\leqslant w\cdot h + E\cdot h(\sum\limits_{k=1}^{k_0}\frac{2^{2^k-1}}{2^k} + 
\sum\limits_{k=k_0+1}^{\ceil{\log_2E}}\frac w{2^k}) - E\cdot h. 
\end{equation}

Using~(\ref{powerpowersum_ineq}) the first sum is bounded by
\begin{equation*}
2^{2^{k_0}-k_0-1} + 2^{2^{k_0-1}-k_0+1}
< 2^{2^t-(t-1)-1} + 2^{2^{t-1}-(t-1)+1} 
= \frac{2^{2^t}}{2^t} + \frac{4\cdot2^{2^{t-1}}}{2^t} =
\frac{2w + 4\sqrt{2w}}{\log_2 w + 1}.
\end{equation*}

The second sum is bounded by 
$$
\sum\limits_{k=k_0+1}^{\infty}\frac w{2^k}=\frac w{2^{k_0}}\leqslant\frac w{2^{t-1}} = \frac{2w}{\log_2w+1}.
$$
After substitution of these terms to~(\ref{tree_assessment_proof_ineq}) we obtain~(\ref{hough_complexity_ineq}).

To prove the depth inequality note that by proposition~\ref{finer_span_part_prop} depth $\depth_{\lset ki}(\lset{k-1}{2i}) = \depth_{\pset ki}(\pset{k-1}{2i})$ (same for $\lset{k-1}{2i+1}$). The bound then follows from~(\ref{main_tree_depth_ineq}) for the $\pset ki$ partition tree (obviously $\card{\pset ki}\leqslant w$).
\end{proof}

\begin{corollary}
On a square $n\times n$ image ($n>1$) the Hough ensemble $\Lens$ generated from $n$ base lines has complexity
\begin{equation}\label{hough_asympt_ineq}
\mu(\Lens) < \frac{4n^3}{\log_2n}\left(1-\frac1{\log_2n+1} + \sqrt{\frac2n}\right).
\end{equation}
The corresponding computation chain has depth at most $\log_2 n\cdot(\log_2 n + 1).$
\end{corollary}

\begin{remark}
Inequalities~(\ref{hough_complexity_ineq}) and~(\ref{hough_asympt_ineq}) hold for any span generated patterns ensemble satisfying~(\ref{hough_consequent_ineq}), because we used only this property to prove the theorem. In fact a weaker condition like $\card{\eqsetens{e}{e+1}}<C$ would also work with the same asymptotics but a different constant.
\end{remark}

\section{Discussion}\label{discussion_section}

\begin{wrapfigure}[6]{R}{0.23\textwidth}
\begin{tikzpicture}
[scale=0.25]
\def\w{15}
\def\e{6}
\pgfmathsetmacro{\q}{\w-1}
\pgfmathsetmacro{\h}{6}

\def\clr{gray!50}

\draw[help lines, dashed, step=1] (0,\h) grid +(\w,2);
\draw[help lines, step=1] (0,0) grid (\w,\h);

\begin{scope}
[xshift=0.5cm, yshift=0.5cm]

\def\sqpix(#1,#2) {
\draw[color=black,fill=\clr] (#1,#2) ++(-0.4,-0.4) rectangle +(0.8,0.8);
}

\foreach \x in {0,...,\q}
{
  \pgfmathtruncatemacro{\y}{round(\e * \x / \q)+1}
  \def\clr{gray!20}
  \sqpix(\x, \y)
}

\foreach \x in {0,...,10}
{
  \def\clr{gray!50}
  \pgfmathtruncatemacro{\y}{mod(round(\e * \x / \q)+1, \h) }
  \sqpix(\x, \y)
}

\foreach \x in {11,...,\q}
{
  \def\clr{gray!20}
  \pgfmathtruncatemacro{\y}{mod(round(\e * \x / \q)+1, \h) }
  \sqpix(\x, \y)
}

\end{scope}
\end{tikzpicture}
\end{wrapfigure}
In the introduction we initially defined the classical Hough transform not for ``cyclic'' (wrapping over image border) lines~(\ref{hough_cyclic_eq}) but for digital line segments~(\ref{intro_lines_eq}). The latter case can always be reduced to the former one by extending the original image vertically (for mostly horizontal lines) and padding it with zero values which changes the computation asymptotics at most by a fixed factor. One can also then prune the redundant (constant zero) input pixels and all the nodes which depend only on them from the generated computation circuit. This justifies the ``cyclic'' $\modfunc{\cdot}$ approach used in sections~\ref{image_section} and~\ref{hough_section}.

The Hough ensemble complexity bound in theorem~\ref{hough_compl_theorem} is {\slshape constructive} as its proof uses an explicit binary partition tree and thus the associated computation chain which is straightforwardly transformed into a computation circuit (section~\ref{connection_section}). Of course, in practice one should not try to directly compute $\lset {k-1}{2i}\comref\lset{k-1}{2i+1}$ in algorithm~\ref{partree_std_alg} but rather use recurrence relation~(\ref{pset_recursive_def_two}) for the $X$-domain.

It is worth noting that applying the proposed algorithm to the fast Hough transform patterns would produce the $n^2 \log_2 n$ asymptotics -- precisely the FHT algorithm complexity! Indeed, consequent FHT patterns (fig.~\ref{fht_fig}) for $w=h=E=n=2^m$ overlap either on $[0,\frac n2)$ or on $[\frac n2, n)$, blocks of patterns with four consequent elevations overlap on segments $[0,\frac n4), \;[\frac n4, \frac n2), \;[\frac n2, \frac{3n}4)$ or $[\frac{3n}4, n)$ and so on, so $\card{\pset ki} = 2^k$, which by~(\ref{main_tree_weight_ineq}) in the same manner as~(\ref{tree_assessment_proof_ineq}) gives
$$
\omega(\thetree)\leqslant n^2 + n^2\sum\limits_{k=1}^m \frac{2^k}{2^k} - n^2 = n^2\log_2 n.
$$
\begin{figure}[ht!]
\centering
\begin{tikzpicture}
[scale=0.33]
\def\w{8}
\def\h{11}

\def\clr{gray!50}
\def\sqpix(#1,#2) {
\fill[fill=\clr] (#1,#2) ++(-0.4,-0.4) rectangle +(0.8,0.8);
}
\def\drawhlines{
    \draw[thick] (4,0)--(4,\h);
    \draw[thin] (2,0)--(2,\h);
    \draw[thin] (6,0)--(6,\h);
}

\begin{scope}
\draw[help lines,step=1] (0,0) grid (\w,\h);
\draw[dotted] (4,0)--(4,-1);
\draw[dotted] (2,0)--(2,-1);
\draw[dotted] (6,0)--(6,-1);
\drawhlines

\begin{scope}
[xshift=0.5cm, yshift=0.5cm]

\foreach \x in {0,...,7}
{
  \sqpix(\x,0)
}

\def\clr{gray!30}
\foreach \x in {0,...,3}
{
  \sqpix(\x,4)
}
\foreach \x in {4,...,7}
{
  \sqpix(\x,5)
}

\end{scope}
\end{scope}


\begin{scope}
[yshift=-12cm]
\draw[help lines,step=1] (0,0) grid (\w,\h);
\drawhlines

\begin{scope}
[xshift=0.5cm, yshift=0.5cm]

\sqpix(0,0)
\sqpix(1,0)
\sqpix(2,1)
\sqpix(3,1)
\sqpix(4,1)
\sqpix(5,1)
\sqpix(6,2)
\sqpix(7,2)

\def\clr{gray!30}
\sqpix(0,5)
\sqpix(1,5)
\sqpix(2,6)
\sqpix(3,6)
\sqpix(4,7)
\sqpix(5,7)
\sqpix(6,8)
\sqpix(7,8)

\end{scope}
\end{scope}


\begin{scope}
[xshift=10cm]
\draw[help lines,step=1] (0,0) grid (\w,\h);
\drawhlines
\draw[dotted] (4,0)--(4,-1);
\draw[dotted] (2,0)--(2,-1);
\draw[dotted] (6,0)--(6,-1);

\begin{scope}
[xshift=0.5cm, yshift=0.5cm]

\sqpix(0,0)
\sqpix(1,1)
\sqpix(2,1)
\sqpix(3,2)
\sqpix(4,2)
\sqpix(5,3)
\sqpix(6,3)
\sqpix(7,4)

\def\clr{gray!30}

\sqpix(0,5)
\sqpix(1,6)
\sqpix(2,6)
\sqpix(3,7)
\sqpix(4,8)
\sqpix(5,9)
\sqpix(6,9)
\sqpix(7,10)

\end{scope}
\end{scope}

\begin{scope}
[xshift=10cm, yshift=-12cm]
\draw[help lines,step=1] (0,0) grid (\w,\h);
\drawhlines

\begin{scope}
[xshift=0.5cm, yshift=0.5cm]

\sqpix(0,0)
\sqpix(1,1)
\sqpix(2,2)
\sqpix(3,3)
\sqpix(4,3)
\sqpix(5,4)
\sqpix(6,5)
\sqpix(7,6)

\def\clr{gray!30}

\sqpix(0,3)
\sqpix(1,4)
\sqpix(2,5)
\sqpix(3,6)
\sqpix(4,7)
\sqpix(5,8)
\sqpix(6,9)
\sqpix(7,10)

\end{scope}
\end{scope}

\end{tikzpicture}
\caption{Pairs of FHT patterns with consecutive elevations: 0-1, 2-3, 4-5, 6-7 (top to bottom, left to right). Vertical lines help distinguish common components. Image width is 8.}
\label{fht_fig}
\end{figure}
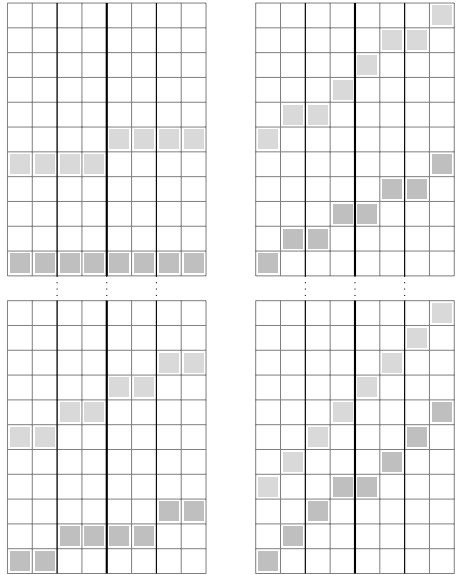
It gives hope that in practice Hough ensemble circuits might yield a much better result than~(\ref{hough_complexity_ineq}). Indeed, the cardinality bound in~(\ref{pset_card_ineq}) is very rough, in practice the $\pset ki$ sets might grow by far not that fast.
Performing the necessary computational experiments as well as assessing the number of operations and memory requirements of algorithm~\ref{partree_std_alg} applied to the Hough ensemble is the plan of our next research.

Another area of interest is generalizing or modifying the suggested approach to handle {\slshape shift-invariant} ensembles consisting of patterns which are not function graphs so their pattern orbits do not per se produce image partitioning.

\end{document}